\documentclass{article}
\usepackage{authblk}
\usepackage{style_stripped_down}
\usepackage{tikz}
\usetikzlibrary{arrows}

\author{Felix Bauckholt}
\author{Kanstantsin Pashkovich}
\author{Laura Sanit\`a}
\affil{Department of Combinatorics and Optimization, University of Waterloo,\\ 200 University Avenue West
Waterloo, ON, Canada N2L 3G1\\
  \texttt{felixbauckholt@gmail.com, kpashkov@uwaterloo.ca, laura.sanita@uwaterloo.ca}}

\title{On the approximability of the stable marriage problem with one-sided ties}

\date{}

\begin{document}
\maketitle

\begin{abstract}
The classical stable marriage problem asks for a 
matching between a set of men and a set of women 
with no blocking pairs, which are pairs formed by a man and a woman
who would both prefer switching from their current status to be paired up together.
When both men and women have strict preferences over the opposite group,
all stable matchings have the same cardinality, and
the famous Gale-Shapley algorithm can be used to find
one. Differently,
if we allow ties in the preference lists, finding
a stable matching of maximum cardinality is an NP-hard problem,
already when the ties are one-sided, that is, they appear only in the preferences of one group.
For this reason, many researchers have focused on developing \emph{approximation algorithm}
for this problem.

In this paper, 
we give a refined analysis of an approximation algorithm 
given by Huang and Telikepalli (IPCO'14) for the stable marriage problem with one-sided ties,
which shows an improved $\frac{13}{9}$-approximation
factor for the problem. Interestingly, our analysis is tight. 
\end{abstract}

\section{Introduction}
The stable marriage problem is one of the most classical and famous algorithmic game theory problems.
We are here given a bipartite graph $G=(A\cup B, E)$, where, according to the widely used terminology, 
$A$ represents a set of men, and $B$ represents a set of women. 
Each individual $a \in A$ (resp. $b \in B$) is adjacent to a subset $N(a) \subseteq B$ (resp. $N(b) \subseteq A$) of individuals,
representing the subset of people (s)he would accept to be matched to.
Furthermore, each $c \in A \cup B$ has strict preferences over the individuals in $N(c)$.
The goal is to compute a \tbdefined{stable matching} $M$, that is a matching which
does not contain a \tbdefined{blocking pair}. 
Formally, a blocking pair~$(a, b)$ for a matching $M$
is given by two individuals $a \in A$ and $b\in B$ with $(a, b) \in E\setminus M$,
such that $a$ is either unmatched in $M$ or he prefers $b$ to his current partner, 
and similarly $b$ is either unmatched in $M$ or she prefers $a$ to his current partner.
Intuitively, if there exists a blocking pair~$(a, b)$ for a given matching $M$,
then both $a$ and $b$ have an incentive to deviate from their current status, making the matching $M$
be an unstable outcome.  The seminal work of Gale and Shapley \cite{GaleShapley62} shows that a stable matching always exists, 
and gives an elegant polynomial-time algorithm to find one. Since then, 
stable matchings have been extensively studied in the literature, and many beautiful properties are known
about their structure. In particular, all stable matchings in a graph have the same cardinality \cite{ruralhospitals2}.

The situation becomes significantly different when the preferences among individuals are not strict anymore, 
and we allow \tbdefined{ties}. In fact, in the presence of ties, 
stable matchings can have different cardinalities, 
and computing a stable matching of maximum size is NP-hard (in fact, APX-hard) 
\cite{incomplete_ties,HardVariants,Iwamaetal07}. 
Due to the importance of finding stable matchings of large size in many real-world applications,
it is not surprising that many researchers in the algorithmic community 
have focused on developing \emph{approximation algorithms} for the stable marriage problem with ties 
\cite{algorithms,Kiraly11,Iwamaetal14,HuangKavitha2015,paluch_original,mcdermid_first_approx,Iwamaetal07}. 

A special case which received a lot of attention, is the stable
marriage with \tbdefined{one-sided} ties, i.e., where only one group
(say, the women) are allowed to have ties in their preferences over
the opposite group. This setting finds many applications in real-world instances, 
such as in the Scottish Foundation Allocation Scheme (SFAS) (we refer to \cite{Irving08}
for details). For this case, Kiraly \cite{Kiraly11} gave a 
$\frac{3}{2}$-approximation bound, later improved to a 
$\frac{25}{17}$-approximation bound by Iwama et al. \cite{Iwamaetal14}. 
Huang and Telikepalli \cite{HuangKavitha2015} developed a new algorithm and showed that their algorithm achieves the approximation factor at most  $\frac{22}{15}$. Subsequently, Radnai~\cite{Radnai} showed that the algorithm of Huang and Telikepalli \cite{HuangKavitha2015} has approximation factor at most  $\frac{41}{28}$.
Dean and Jalasutram \cite{BJ15} improved the analysis of the algorithm proposed by 
Iwama et al. \cite{Iwamaetal14} to $\frac{19}{13}$, and this has been the best known
approximation factor for the problem until recently.

From an hardness point of view, the strongest bounds
are given in \cite{Halldorsson:2007:IAR:1273340.1273346}, which
showed that the stable marriage problem with one-sided ties is NP-hard to approximate
within a factor of $\frac{21}{19}$, and strengthened the bound to $\frac{5}{4}$
assuming Unique Game Conjecture.

\paragraph{Our results and techniques.} 
In this paper, we give a refined analysis of the 
algorithm of Huang and Telikepalli \cite{HuangKavitha2015}, 
which shows that the algorithm is indeed a $\frac{13}{9}$-approximation
algorithm for the problem. Note that the factor $\frac{13}{9} \approx 1.444$
is strictly better than the factor $\frac{19}{13} \approx 1.461$ given in \cite{BJ15}.

An interesting feature of the combinatorial algorithm of \cite{HuangKavitha2015} is that, 
differently from the classical Gale-Shapley procedure, 
each man can make two proposals and each woman can
accept two proposals. When men stop proposing, 
the set of accepted proposals yields a graph $G'$ where each individual
can have degree up to 2, and the matching $M$ output by the algorithm
will be a maximum cardinality matching in $G'$ which satisfies some additional properties. 

Given a stable matching $M$ output by their algorithm and a 
stable matching $\opt$ of maximum cardinality, 
the idea of \cite{HuangKavitha2015} for proving the approximation
bound is the following. 
They consider
the $M$-augmenting paths in the symmetric difference ($M\oplus \opt$)
of $M$ and $\opt$. 
It is not difficult to show that $M$-augmenting paths of size 3 (i.e. with 3 edges) cannot exist, 
if the algorithm is executed properly. 
However, there could be an augmenting path of size 5 or larger.
In particular, augmenting paths of size 5 are problematic, since for such paths
the ratio between the number of the edges in $\opt$ and $M$ is $3/2$, and therefore exceeds
the claimed approximation bound.
Huang and Telikepalli \cite{HuangKavitha2015} are able to \emph{charge} each augmenting path of size 5 to some component of 
$M\oplus \opt$ of a larger size. The crux of the analysis lies in defining a map 
between augmenting paths of size 5, and 
large components of $M\oplus \opt$,
such that each large component 
does not get charged too many times.
In order to define this map, 
they study the structure of an \emph{auxiliary} directed graph $H$ which
is obtained from $G'$ by orienting its edges and contracting some
edges of $M$ which belong to $5$-augmenting paths.

Our improvement is based on defining 
a better mapping between augmenting paths of size 5, and 
large components of $M\oplus \opt$. In order to do this, we 
construct a slightly different auxiliary directed graph $H$, 
where we contract more edges of $M$ (not just the ones in 5-augmenting paths as is done in~\cite{HuangKavitha2015}).
We study the structure of this new graph, and introduce two crucial concepts:
\emph{popular women}, and \emph{jumps}. In particular, jumps
are central to our refined charging scheme, and together with 
the concept of popular women, we are able to provide a better analysis of the algorithm.

Interestingly, an example by Radnai~\cite{Radnai} shows that our analysis is tight. For the sake of completeness, also in this paper we report an instance
where the algorithm effectively computes a solution which is away from the 
optimal one by a factor of $\frac{13}{9}$.  
This shows that our charging scheme is essentially best possible,
and modifications in the algorithm are necessary in order to further improve the approximation bound.

\smallskip
\noindent{\bf Remark:} Very recently, Lam and Plaxton \cite{LP18} have posted online a paper
showing a $\ln(4) \approx 1.386$ algorithm for the problem. This is achieved by improving
the algorithm and technique of \cite{Iwamaetal14} and \cite{BJ15}.
The work of Lam and Plaxton \cite{LP18} was done indipendently from our work. 
Though their factor is strictly better
than ours, their algorithm involves solving an LP, while 
the algorithm of Huang and Telikepalli \cite{HuangKavitha2015}, that we use here, is a purely combinatorial algorithm,
which is one of the nicest aspects of it. Therefore our result is of independent interest.

\section{Algorithm by Huang and Telikepalli}
\label{section:algorithm}

First, let us describe the algorithm by Huang and Telikepalli~\cite{HuangKavitha2015}. 
A stable matching is computed in two phases. In the first phase, which is the proposal
phase, men (in arbitrary order) make proposals to women, 
while women accept or reject proposals. It is important to highlight that at any point of time
each men can simultaneously make up to two proposals, and women are allowed to accept
up to two proposals. This phase stops when all men stop making proposals.
Subsequently, one computes a maximum cardinality matching with some properties in the graph induced by all the accepted proposals. In this graph,
every node can have degree at most two (since both men and women are allowed to make/accept up to two proposals). 
We describe the process  in more details in the next paragraphs.

\subsubsection*{How men propose.} Each man $a\in A$ has two proposals $p^1_a$ and $p^2_a$. Initially, both $p^1_a$ and $p^2_a$ are offered to the first woman on $a$'s list. If a proposal $p^i_a$, $i=1,2$ is rejected by the $k$-th woman on $a$'s list, then the proposal $p^i_a$ goes to the $(k+1)$-st woman on $a$'s list; in case the $k$-th woman is the last woman on $a$'s list, then the proposal $p^i_a$ goes again to the first woman on $a$'s list.

At each moment of the algorithm every man has one of the following three statuses: \tbdefined{basic}, \tbdefined{$1$-promoted} and \tbdefined{$2$-promoted}. Each man $a\in A$ starts as a basic man. If each woman in $N(a)$ at least once rejected any proposal of $a$, the man $a$ becomes $1$-promoted. If afterwards each woman in $N(a)$ at least once rejected any proposal of $a$ as a $1$-promoted man, the man $a$ becomes $2$-promoted. Finally, if each woman in $N(a)$ at least once rejected any proposal of $a$ as a $2$-promoted man, the man $a$ stops making proposals.

\subsubsection*{How women accept proposals.}
 A woman $b$, who gets a proposal from a man $a$, always accepts the current proposal of $a$ if at the moment $b$ holds at most one proposal, excluding the current proposal of $a$. Otherwise, if $b$ holds already two proposals when she receives the proposal of $a$, 
 then $b$ rejects a \tbdefined{least desirable proposal} among these three 
 (see Definition~\ref{def:proposals_preferences}), and keeps the other two (if more than one proposal 
 among the ones considered satisfies Definition 1, ties are broken arbitrarily). 

\begin{definition}
\label{def:proposals_preferences}
For a woman $b$, proposal $p^i_a$ is superior to $p^{i'}_{a'}$ if one of the following is true
\begin{itemize}
\item $b$ prefers $a$ to $a'$.
\item $b$ is indifferent between $a$ and $a'$; $a$ is currently $2$-promoted while $a'$ is not $2$-promoted.
\item $b$ is indifferent between $a$ and $a'$; $a$ is currently $1$-promoted while $a'$ is basic.
\item $b$ is indifferent between $a$ and $a'$; $a$ and $a'$ are basic; moreover, woman $b$ has already rejected a proposal of $a$ while $b$ did not reject any proposal of $a'$. 
\end{itemize}
A proposal $p^i_a$ is a \tbdefined{least desirable proposal} among a set of proposals that a woman $b$ has, if
it is  not superior to any of the proposals in the set.
\end{definition} 
 
\subsubsection*{What is the output matching.}

Let $G'$ be the bipartite (simple) graph with the node set $A\cup B$ and the edge set $E'$, where $E'$ consists of the edges $(a,b)$, $a\in A$ and $b\in B$ such that at the end of the algorithm $b$ holds at least one proposal of $a$. 
Clearly, the degree of a node in $G'$ is at most two, since each man has at most two proposals and each woman has at most two proposals accepted at each time point. Let $M$ be a maximum cardinality matching in $G'$, where all degree two nodes of $G'$ are matched.
In~\cite{HuangKavitha2015}, it was shown that the matching $M$ is a stable matching in the graph $G$.

\subsubsection*{Useful properties and notations.}
Let $\opt$ be a maximum cardinality stable matching in $G$. Then $M\oplus \opt$ consists of paths and cycles. We will say that a path is an \tbdefined{augmenting} path in $M\oplus \opt$ if it is augmenting with respect to $M$.
Unless explicitly stated otherwise, a man $a\in A$ is called \tbdefined{basic}, \tbdefined{$1$-promoted} and \tbdefined{$2$-promoted} if after the end of the algorithm $a$ is basic, $1$-promoted or $2$-promoted, respectively. A man $a\in A$ is called \tbdefined{promoted} if $a$ is either $1$-promoted or $2$-promoted.

The following two lemmas are from~\cite{HuangKavitha2015}, and give some useful properties of augmenting paths in $M\oplus \opt$. 
For the completeness sake we provide the proofs of these two lemmas in Appendix.

\begin{lemma}[\cite{HuangKavitha2015}]
\label{lemma:three_path_original_paper}
There is no $3$-augmenting path in $M\oplus \opt$.
\end{lemma}

\begin{lemma}[\cite{HuangKavitha2015}]
\label{lemma:five_path_original_paper}
Let $\alpha_0 - \beta_0 - \alpha_1 -  \beta_1 - \alpha_2  - \beta_2$ be a $5$-augmenting path in $M\oplus \opt$, where $\alpha_i\in A$, $\beta_i\in B$ for $i=0,1,2$, then
\begin{itemize}
\item $\alpha_0$ is a $2$-promoted man and has been rejected by $\beta_0$ as a $2$-promoted man.
\item $\alpha_2$ is basic and $\alpha_2$ prefers $\beta_1$ to $\beta_2$.
\item $\alpha_1$ is not $2$-promoted and $\alpha_1$ prefers $\beta_1$ to $\beta_0$.
\item $\beta_1$ is indifferent between $\alpha_1$ and $\alpha_2$.
\item In $G'$, $\beta_0$ has degree $1$ if and only if $\alpha_1$ has degree $1$.
\item In $G'$, $\beta_1$ has degree $1$ if and only if $\alpha_2$ has degree $1$.
\end{itemize}
\end{lemma}

Further in the paper, we use $\leq$, $\geq$, $>$, $<$, $\indiff$ to indicate preferences. For example, $a>_b c$ stands for $b$ strictly prefers $a$ to $c$. For a node $a$ in $G$, $M(a)$ and $\opt(a)$ stand for the node matched with $a$ by $M$ and $\opt$, respectively. If  $a$ is not matched by $M$ or $\opt$, we define $M(a)$ and $\opt(a)$ to be $\varnothing$, respectively.

\section{Tight Analysis}
\label{sec:tight_analysis}

As in \cite{HuangKavitha2015}, we will now construct a digraph $H$, starting from $G'$, which will be used for our charging scheme.
First, we orient all edges in the graph $G'$ from men to women.  Second, for each augmenting path of the form $\alpha_0 - \beta_0 - \alpha_1 - \dots - \alpha_k - \beta_k$ with $\alpha_i\in A$ and $\beta_i\in B$ for all $i=0,1,\ldots,k$, we contract the edge $(\alpha_1, \beta_0)$ into a single node (we call the resulting node an \tbdefined{$x$-node}) and we contract the edge $(\alpha_k, \beta_{k-1})$ into a single node (we call the resulting node a \tbdefined{$y$-node}). Note that this second step in the construction of the graph $H$ is different from the construction in~\cite{HuangKavitha2015}, which is based only on augmenting paths of length~$5$. We refer to the $x$- and $y$-nodes of $H$ as \tbdefined{red}; and to the other nodes of $H$ as \tbdefined{blue}.
For every node $v \in V(G')$, let $[v] \in V(H)$ be the corresponding (possibly an $x$- or a $y$-) node in $H$. We say that a node $v\in V(G')$ is \tbdefined{in an $x$-node} or \tbdefined{in a $y$-node} if  $[v]$ is an $x$-node or a $y$-node, respectively.
A crucial concept is that of critical arcs and good paths, defined below.

\begin{definition}
An arc $([a], [b])$ in $H$ with $a\in A$ and $b\in B$, is called  \tbdefined{critical} if all statements below hold
\begin{itemize}
\item $b \geq_a \opt(a)$.
\item $a$ is not $2$-promoted.
\item $(a, b) \notin M$.
\item $[a]$ is not a $y$-node.
\end{itemize}
\end{definition}

\begin{definition}
A (directed) path in $H$ is called \tbdefined{good} if it starts with a blue man, ends with a blue woman, and no edge of the path is in $M$.
\end{definition}

We now state a lemma which gives certain properties of the graph $H$. These properties are analogous to the properties
proven in~\cite{HuangKavitha2015}. However, since our definition of $H$ and our definition of a critical arc is different from the definitions in~\cite{HuangKavitha2015}, we provide a proof of this lemma in Appendix.

\begin{lemma}
\label{lemma:combined}
The graph $H$ satisfies the following statements.
\begin{itemize}
\item[(a)] \label{lemma:no_yx_arc} There is no arc from a $y$-node to an $x$-node.
\item[(b)] \label{lemma:critical_arc} No critical arc starts at a $y$-node and no critical arc ends at an $x$-node.
\item[(c)] \label{lemma:critical_arc_is_in_some_good_path} Every critical arc is contained in some good path.
\item[(d)] \label{lemma:good_path_has_at_most_one_critical_arc} Every good path has at most one critical arc.
\item[(e)] \label{lemma:good_paths_are_node_disjoint} Every two distinct good paths are node-disjoint.
\end{itemize}
\end{lemma}

%

%

%

%

%

\subsection*{Popularity}
In order to obtain a better charging scheme, 
we here introduce the concept of \tbdefined{popular women}. This concept plays a key role later in our analysis.

\begin{definition}
At a given time point, a proposal from $c\in A$ to $b\in B$ is \tbdefined{$a$-good} (for $a\in A$) if $c\geq_b a$, and at least one of the following is true
\begin{itemize}
\item $c$ is promoted at the given time point
\item $c$ was rejected by $b$ before the given time point
\item $c >_b a$.
\end{itemize}
\end{definition}

\begin{definition}
For $a\in A$ and $b\in B$, $b$ is \tbdefined{$a$-popular} if
at the end of the algorithm $b$ holds two $a$-good proposals.
\end{definition}

The following lemma is a straightforward observation.

\begin{lemma}
\label{lemma:popularity_monotonicity_preference}
Let $a,c\in A$ and $b\in B$ be such that $a\geq_b c$. If $b$ is $a$-popular then $b$ is $c$-popular as well.
\end{lemma}

The next lemmas will be useful later.

\begin{lemma}
\label{lemma:monotonicity_popularity_over_time}
For $a\in A$ and $b\in B$, the number of $a$-good
proposals which $b$ holds does not decrease over time.
Furthermore, if $b$ is not $a$-popular, $b$ rejected $a$ at most one time.
\end{lemma}

\begin{proof}
First, note that a proposal from $c\in A$ to $b\in B$ that is not $a$-good for some $a\in A$ at some point in time, can become $a$-good 
at some later point in time. For example, this can happen when $c$  becomes promoted. 
However, once a proposal becomes $a$-good, it will remain so until the end of the algorithm. 
To complete the proof of the first statement, it remains to show that at no point in time $b$ rejects an $a$-good proposal for a not $a$-good proposal. Let us assume the contrary, i.e. that at some point in time $b\in B$ rejects an $a$-good proposal from $c\in A$ for a not $a$-good proposal from $d\in A$. We have $c\geq_b a$ and $a \geq_b d$, since the proposal from $c$ is $a$-good and the proposal from $d$ is not. We also have $d\geq_b c$, since $b$ rejects the proposal of $c$ for the proposal from $d$. Thus, we have $a \indiff_b c \indiff_b d$.
Moreover, since the proposal from $d$ is not $a$-good, $d$ is not promoted and was not rejected by $b$ before the given time point. However, since $b$ rejects the proposal from $c$ for the proposal from $d$,  at this time point $c$ is also neither promoted nor was rejected by $b$ before. Hence, the proposal from $c$ is not $a$-good, contradiction.

For the second statement, note that if $b\in B$ rejects $a\in A$ once, then all subsequent proposals of $a$ to $b$ are $a$-good. Thus, due to the above argumentation, if at some time point $b$ rejects $a$ for the second time, then at that time point $b$ holds two $a$-good proposals. Hence, at the end of the algorithm $b$ holds two $a$-good proposals as well, i.e. $b$ is $a$-popular.
\end{proof}

\begin{lemma}
\label{lemma:least_prefered_unpopular}
Let $a\in A$ be basic, then the less preferred neighbor of $a$ in $G'$ is not $a$-popular.
\end{lemma}

\begin{proof}
Let $b\in B$ be the less preferred neighbor of $a$ in $G'$. It is sufficient to show that at the end of the algorithm, $b$ holds a proposal that is not $a$-good. Note, $b$ did not reject $a$ at any time point. Moreover, since $a$ is basic, any proposal of $a$ to $b$ is not $a$-good.
\end{proof}

\begin{definition}
For a woman $b\in B$, a critical arc $([a], [c])$ is \tbdefined{next to $b$} if $(a,b)$ is in $\opt$.
\end{definition}

\begin{lemma}
\label{lemma:critical_arc_five_path}
If $\alpha_0 - \beta_0 - \alpha_1 - \beta_1 - \alpha_2 - \beta_2$ is a $5$-augmenting path and $\beta_1$ is not $\alpha_1$-popular, then there is a critical arc next to $\beta_1$.
\end{lemma}

\begin{proof}

Let us assume that for all $b\in B$ such that $(\alpha_1,b)$ is an arc in $G'$, we have $\beta_1>_{\alpha_1} b$. Then $\beta_1$ rejected $\alpha_1$ at least twice, contradicting Lemma~\ref{lemma:monotonicity_popularity_over_time}.

Thus, there exists $b\in B$ such that $(\alpha_1,b)$ is an arc in $G'$, we have $b\geq_{\alpha_1} \beta_1=\opt(\alpha_1)$. Let us prove that $([\alpha_1],[b])$ is a critical arc in $H$. By Lemma~\ref{lemma:five_path_original_paper}, $\alpha_1$ is not $2$-promoted and $\beta_1 >_{\alpha_1} \beta_0=M(\alpha_1)$. Clearly, $[\alpha_1]$ is not a $y$-node and $(\alpha_1, b)\not\in M$, proving that  $([\alpha_1],[b])$ is a critical arc in $H$. 
\end{proof}

\subsection*{Jumps}
The second concept that will be central for our improved analysis, is that of \tbdefined{jumps}. 
We start with introducing \tbdefined{matching jumps}. Each such jump defines a map from a woman, which is not isolated in $G'$ but is not matched by $M$, to a uniquely defined woman, which is matched by~$M$. 

\begin{definition}
Let $b\in B$ be  not matched by $M$, such that $b$ is not an isolated node in $G'$. Then there exists $c\in B$ such that $b$ and $c$ are end nodes of a maximal (not directed) path in $G'$. Let us define $\tbdefined{\mjump b}:= c$.
\end{definition}

\begin{remark}
Let $b\in B$ be  not matched by $M$, such that $b$ is not an isolated node in $G'$. Then $b$ has degree $1$ in the graph $G'$. Thus there is a unique maximal (not directed) path in $G'$ such that $b$ is an end node of this path. By the construction of $M$, the other end node of this path is some woman $c\in B$ who is matched by $M$. 
\end{remark}

Let us prove that for two different women $b$, $c\in B$ their matching jumps cannot result in the same woman, whenever matching jumps are well defined for $b$ and $c$.

\begin{lemma}
\label{lemma:mjump_injective}
Let $b,c\in B$ be not matched by $M$, such that $b$ and $c$ are not isolated nodes in $G'$.  If $\mjump b = \mjump c$, then $b = c$.
\end{lemma}

\begin{proof}
Note that $b$ and $\mjump b$ are two end nodes of the same maximal (not directed) path in $G'$. Similarly, $c$ and $\mjump c$ are two end nodes of the same maximal (not directed) path  in $G'$. Since maximal (not directed) paths in the graph $G'$ are disjoint and $\mjump b = \mjump c$, we have $b = c$.
\end{proof}

\begin{lemma}
Let $b\in B$ be not matched by $M$, such that $b$ is not an isolated node in $G'$. Then, the node $\mjump b$ is either a blue woman in the graph $H$ or is in a $y$-node.
\end{lemma}
\begin{proof}
It is enough to show that $\mjump b$ is not in an $x$-node. Note that if a woman $c\in B$ is in an $x$-node, then this woman holds two proposals at the end of the algorithm. Thus, $c$ either has degree $2$ in the graph $G'$ or $(c, M(c))$ forms a connected component of $G'$ (that is,
both two proposals held by $c$ are from $M(c)$). However, by the definition $\mjump b$  is an end node of a maximal (not directed) path in $G'$, where one of the end nodes is a woman not matched by $M$, contradiction.
\end{proof}

As follows from the above lemma, a matching jump can result in a blue node or in a $y$-node. However, for reasons to become clear later, we would like jumps to map to a blue woman. Sometimes a matching jump indeed ends in a $y$-node.   In such cases, the next type of jumps, \tbdefined{path jumps},  helps to do a further jump to a uniquely defined blue woman. 

\begin{definition}
Let $c$ be a $y$-node or an end node of a critical arc in $H$, let $b$ be a blue woman such that there is a directed path in $H$ from $c$ to $b$. Then $\tbdefined{\pathjump c}:= b$.
\end{definition}

\begin{remark}
Note that if a node $c$ in $H$ is in a critical arc, then $\pathjump c$ is uniquely defined, and $\pathjump c$ is the end node of the good path containing $c$.

Moreover, $\pathjump c$ is also uniquely defined for $[c]$, whenever $[c]$ is a $y$-node with an outgoing arc. If $c$ has no outgoing arc, then there is no blue woman $b$ such that there is a directed path in $H$ from $c$ to $b$.
\end{remark}

Finally, we introduce a map \tbdefined{matching jump with exception}, which combines the previous two jumps.

\begin{definition}
Let $b\in B$ be not matched by $M$, such that $b$ is not an isolated node in $G'$. 
If $\mjump b$ is a blue woman, let 
$$\tbdefined{\mjumpe b} := \mjump b\,.$$
If $\mjump b$ is in a $y$-node, let 
$$\tbdefined{\mjumpe b} :=\pathjump{[\mjump b]}\,.$$
\end{definition}

\begin{remark}
Let $b\in B$ be not matched by $M$, such that $b$ is not an isolated node in $G'$.  If $\mjump b$ is in a $y$-node, then this $y$-node $[\mjump b]$ has an outgoing arc.
\end{remark}

\begin{lemma}
\label{lemma:mjumpe_blue}
Let $b\in B$ be not matched by $M$, such that $b$ is not an isolated node in $G'$. Then $\mjumpe b$ is a blue woman.
\end{lemma}

\begin{proof}
If $\mjump b$ is a blue woman, then $\mjumpe b = \mjump b$ and hence $\mjumpe b$ is a blue woman as well.
Otherwise $\mjump b$ is in a $y$-node, then $$\mjumpe b = \pathjump{[\mjump b]}$$ and $\pathjump{[\mjump b]}$ is a blue
woman by definition.
\end{proof}

\begin{lemma}
\label{lemma:mjumpe_degree}
Let $b\in B$ be a woman such that $b$ is not matched by $M$ and $b$ is in a good path (so $b$ is not isolated in $G'$). 
If $\mjump b$ is in a $y$-node, then $\mjumpe b$ has degree $2$ in $G'$.
\end{lemma}

\begin{proof}
Suppose  $\mjumpe b$ does not have degree $2$ in $G'$. Then the degree of $\mjumpe b$ in $G'$ equals $1$, since all nodes in $G'$ have degree at most $2$ and the node $\mjumpe b$ is clearly not isolated.

Since $b$ is not matched by $M$, the degree of $b$ in $G'$ is at most $1$. Moreover, the degree of $b$ in $G'$ equals $1$, because $b$ is contained in a good path in the graph $H$. Hence, there exists a unique maximal (not directed) path in $G'$ with an end node $b$. The other end node of this path is $\mjump b$. Note that $\pathjump{[\mjump b]}$ is a node on this path. Since $\pathjump{[\mjump b]}=\mjumpe b$ has degree $1$ in the graph $G'$, the node $\mjumpe b$ is also an end node of this path. Thus, we have that either $\mjump b=\mjumpe b$ or $b=\mjumpe b$. By definition  $\mjumpe b$ is a blue node and by the statement of the lemma $\mjump b$ is in a $y$-node, implying $b=\mjumpe b$.

By definition, there exists a directed path in the graph $H$ from the node $[\mjump b]$ to the node $\pathjump{\mjump b}$. Thus, on one side $[\mjump b]$ is in the good path in $H$ containing the node $b$. On the other side, by definition $[\mjump b]$ has no incoming arc in the graph~$H$. However, no good path contains a $y$-node with no incoming arcs, contradiction. 
\end{proof}

\begin{corollary}
\label{corollary:mjumpe_matched}

Let $b\in B$ be a woman such that $b$ is not matched by $M$ and $b$ is in a good path (so $b$ is not isolated in $G'$).
Then $\mjumpe b$ is matched by $M$.
\end{corollary}
\begin{proof}
If $\mjumpe b=\mjump b$, then by definition $\mjumpe b$ is matched by $M$. If $\mjump b$ is in a $y$-node, then by Lemma~\ref{lemma:mjumpe_degree} $\mjumpe b$ has degree $2$ in the graph $G'$, and hence by definition of $M$ the node $b$ is matched by $M$.
\end{proof}

\begin{lemma}
\label{lemma:mjumpe_injective}
Let $b,c\in B$ be women such that $b$, $c$ are not matched by $M$ and both $b$ and $c$ are in some good path (so neither $b$ or $c$ is isolated in $G'$). If $\mjumpe b = \mjumpe c$, then $b = c$.
\end{lemma}

\begin{proof}
Since all nodes in $G'$ have degree at most $2$ and $\mjumpe b$ is not isolated in $G'$, the degree of $\mjumpe b$ in the graph $G'$ is either $1$ or $2$.

Let us assume that the degree of $\mjumpe b$ in the graph $G'$ is $1$. Then by Lemma~\ref{lemma:mjumpe_degree}, neither $\mjump b$ nor $\mjump c$ is a $y$-node. Thus,
$$
\mjump b=\mjumpe b=\mjumpe c=\mjump c\,,
$$
that by Lemma~\ref{lemma:mjump_injective} implies $b=c$.

Now, let us assume that the degree of $\mjumpe b$ in the graph $G'$ is $2$. Since both $\mjump b$ and $\mjump c$ have degree $1$ in $G'$, we have
$$
\mjump b\neq\mjumpe b\qquad \mjump c\neq\mjumpe c\,,
$$
implying that both $\mjump b$ and $\mjump c$ are in $y$-nodes with no incoming arcs in $H$. Thus,  from both $[\mjump b]$ and $[\mjump c]$ there is a directed path to the node $\mjumpe b=\mjumpe c$ in the graph $H$, where all nodes in the path except $\mjumpe b$ are $y$-nodes. On the other hand, $\mjumpe b$ is a blue woman in $H$, and hence there is at most one arc from a $y$-node to $\mjumpe b$. Moreover, for every $y$-node there is at most one incoming arc, implying $\mjump b=\mjump c$. So in this case, we also obtain $b=c$ by Lemma~\ref{lemma:mjump_injective}.

\end{proof}

The next lemma is used to guarantee that at certain stages of our charging scheme the same nodes do not get charged multiple times.

\begin{lemma}
\label{lemma:pathjump_injective}
Let $a,c\in A$ and $b,d\in B$ be such that the arcs $([a], [b])$ and $([c], [d])$ are critical. If $\pathjump{[a]} =
\pathjump{[c]}$, then we have $a = c$.
\end{lemma}

\begin{proof}
By Lemma~\ref{lemma:critical_arc_is_in_some_good_path}, both $([a],[b])$ and $([c], [d])$ are contained in some good paths. Moreover, $\pathjump{[a]}$ and $\pathjump{[c]}$ are the end nodes of those good paths. Since by Lemma~\ref{lemma:good_paths_are_node_disjoint} good paths are node-disjoint, $([a],[b])$ and $([c], [d])$ are in the same good path. Moreover, since by Lemma~\ref{lemma:good_path_has_at_most_one_critical_arc} every
good path has at most one critical arc, we have $([a],[b])=([c], [d])$, in particular $a = c$.
\end{proof}

\section{Charging scheme}
\label{sec:charging}

To show the tight approximation guarantee of the algorithm in~\cite{HuangKavitha2015}, we use the following charging scheme. The charging scheme is conducted in five stages, which are described below.

\begin{enumerate}
\item \label{step:charging_1} Each man who is in a $y$-node from a $5$-augmenting path receives a charge of $1$.
\item \label{step:charging_2} Each man with a nonzero charge passes on his charge to his less preferred neighbor in~$G'$.
\item \label{step:charging_3}
Every woman $b\in B$ with a nonzero charge and who is not matched by $M$, passes all her charge  to $\mjump b$. (Note that if a woman has a nonzero charge at the beginning of this stage, then she is not isolated in $G'$, and hence this step is well defined.)
\item \label{step:charging_4}
Each woman $b\in B$, such that there is a critical arc $(a, c)$ next to $b$, passes on all her charge to $\pathjump a$. (Note that $\pathjump a$ equals $\pathjump{c}$.)
\item \label{step:charging_5}
Each  woman $b\in B$, who received a nonzero charge in stage \ref{step:charging_4} and who is not matched by $M$, passes on all her charge to $\mjumpe b$.
(Note that if a woman received charge at stage \ref{step:charging_4}, then she is not isolated in $G'$, and hence this stage is well defined.)
\end{enumerate}

We call the charge a woman received at stages \ref{step:charging_2} and \ref{step:charging_3} \tbdefined{unpopularity-charge} and call the charge a woman received at stages \ref{step:charging_4}
and \ref{step:charging_5} \tbdefined{path-charge}.
If a woman receives charge at stages \ref{step:charging_2} and \ref{step:charging_3}, but passes on that
charge at stage \ref{step:charging_4}, we say that this woman has zero unpopularity-charge.

\begin{lemma}
\label{lemma:only_matched_women_have_charge}
At the end of the charging scheme, only women matched by $M$ have a nonzero charge.
\end{lemma}

\begin{proof}
From the charging scheme, it is clear that only women may have a nonzero charge at the end of the charging scheme.

Let us discuss how a woman $c\in B$, who  is not matched by $M$, may receive a nonzero charge during the charging scheme. First, $c$ cannot receive a nonzero charge in stage~\ref{step:charging_1}. Second, if $c$ receives a charge at stage~\ref{step:charging_2}, then $c$ passes on this charge at stage~\ref{step:charging_3}. Third, note that for every $b\in B$ which is not matched by $M$ and is not isolated in $G'$, $\mjump b$ is matched by $M$. Hence, $c$ cannot receive any charge at stage~\ref{step:charging_3}. Forth, if $c$ receives a charge at stage~\ref{step:charging_4}, then $c$ passes this charge on at stage~\ref{step:charging_5}. Fifth, by Corollary~\ref{corollary:mjumpe_matched} and Lemma~\ref{lemma:critical_arc_is_in_some_good_path} no woman, which is not matched by $M$, gets a nonzero charge at stage~\ref{step:charging_5}. 
\end{proof}

\begin{lemma}
\label{lemma:only_blue_women_have_path_charge}
At the end of the charging scheme, only blue women have a nonzero path-charge.
\end{lemma}

\begin{proof}
Note that $\pathjump b$ is a blue woman for every $b\in B$ that passes on her charge at stage~\ref{step:charging_4}. Thus, only blue women can receive a nonzero charge at stage~\ref{step:charging_4}.

For every $b\in B$ that passes on her charge at stage~\ref{step:charging_5},
$\mjumpe b$ is a blue woman by Lemma~\ref{lemma:mjumpe_blue}. Thus, only blue women can receive a nonzero charge at stage~\ref{step:charging_5}.
\end{proof}

\begin{lemma}
\label{lemma:x_node_no_unpopularity_charge}
At the end of the charging scheme, no woman in an $x$-node has a nonzero unpopularity-charge.
\end{lemma}

\begin{proof}
Let us assume that there exists a woman $c\in B$, who is in an $x$-node and holds a nonzero unpopularity-charge at the end of the charging scheme. By Lemma~\ref{lemma:no_yx_arc}, in $G'$ the node $c$ is adjacent to no man in a $y$-node. Thus, $c$ does not receive any charge at stage~\ref{step:charging_2}.

Thus, $c$ received a nonzero charge at stage~\ref{step:charging_3}. So $c$ is an end node of a (not directed non trivial) maximal path in $G'$, where both end nodes are distinct women. However, since  $c$ is in an $x$-node then either $c$ has degree $2$ in $G'$ or $(c, M(c))$ forms a connected component of~$G'$, contradiction.
\end{proof}

\begin{lemma}
\label{lemma:zero_total_charge_five_augmenting_path}
At the end of the charging scheme, no woman in a $5$-augmenting path has a nonzero charge.
\end{lemma}

\begin{proof}
Let $\alpha_0 - \beta_0 - \alpha_1 - \beta_1 - \alpha_2 - \beta_2$
be a $5$-augmenting path, where $\alpha_0$, $\alpha_1$, $\alpha_2 \in A$ and $\beta_0$, $\beta_1$, $\beta_2 \in B$.

First, $\beta_2$ is not matched by $M$ and thus has no charge at the end of the charging scheme by Lemma~\ref{lemma:only_matched_women_have_charge}.

Second, $\beta _0$ has no unpopularity-charge by Lemma \ref{lemma:x_node_no_unpopularity_charge} and no path-charge by Lemma~\ref{lemma:only_blue_women_have_path_charge}, since $\beta_0$ is in an $x$-node.

Finally, $\beta_1$ also has no path-charge by Lemma~\ref{lemma:only_blue_women_have_path_charge}, since $\beta_1$ is in a $y$-node. To finish the proof, it remains to show that at the end of the charging scheme $\beta_1$ has no unpopularity-charge.

Let us assume that $\beta_1$ receives a nonzero unpopularity charge at stage~\ref{step:charging_3}. Then $\beta_1=\mjump b$ for some $b\in B$, where $b$ is not matched by $M$ and is not isolated in $G'$. Thus, $\beta_1$ is the end node of a maximal (not directed) path in $G'$ containing at least two distinct women. However, the degree of $\beta_1$ in $G'$ equals the degree of $\alpha_2$ in $G'$ by Lemma~\ref{lemma:five_path_original_paper}, and so $\beta_1$ cannot be the end node of a maximal (not directed) path in $H'$ containing at least two distinct women.

Let us assume that $\beta_1$ receives a nonzero unpopularity charge at stage~\ref{step:charging_2}. Thus there exists $a\in A$, such that $a$ is in a $y$-node  and $\beta_1$ is the less preferred neighbor of $a$ in $G'$. Since $a$ is in a $y$-node, $a$ is basic and hence $\beta_1$ is not $a$-popular by Lemma~\ref{lemma:least_prefered_unpopular}.

 Since $a$ is in a $y$-node, $\opt(a)$ is not matched by $M$ and $\beta_1 \geq_a \opt(a)$. Note that $\beta_1$ is distinct from $\opt(a)$, since $\beta_1$ is matched by $M$, implying $\beta_1 >_a \opt(a)$. Thus, we have $\alpha_1\geq_{\beta_1} a$, because otherwise $\opt$ is not a stable matching in $G$. By Lemma~\ref{lemma:popularity_monotonicity_preference},  $\alpha_1\geq_{\beta_1} a$ implies that $\beta_1$ is not $\alpha_1$-popular. Hence, by Lemma~\ref{lemma:critical_arc_five_path} there is a critical arc next to $\beta_1$ and so $\beta_1$ passes on her unpopularity-charge at stage~\ref{step:charging_4}. 
\end{proof}

\begin{lemma}
\label{lemma:upper_bound_unpopularity_charge}
After stages~\ref{step:charging_2} and~\ref{step:charging_3}, every woman has a charge of at most $2$. Furthermore, after stages~\ref{step:charging_2} and~\ref{step:charging_3},  each woman, who is not matched by $M$ with a man from a $5$-augmenting path,  has a charge of at most $1$.
\end{lemma}

\begin{proof}
Every woman has at most $2$ neighbors in $G'$. Thus, each woman has a charge of at most $2$ after stage~\ref{step:charging_2}.

Let us assume that a woman $c\in B$ receives a charge at stage~\ref{step:charging_3} passed from some woman $b\in B$. Then, both $b$ and $c$ are the end nodes of a maximal (not directed) path in $G'$. Hence, both $b$ and $c$ have degree $1$ in $G'$, and thus both $b$ and $c$ have a charge of at most $1$ after stage~\ref{step:charging_2}. So after stage~\ref{step:charging_3}, $c$ has a charge of at most $2$.

Let us show the second part of the statement. If $c$ is a woman not matched by $M$, then it is straightforward to see that $c$ receives a charge of at most $1$ at stage~\ref{step:charging_2} and receives no charge at stage~\ref{step:charging_3}. Now, suppose that $(a, c) \in M$ for some $a\in A$ and $c\in B$, where $a$ is not in a $5$-augmenting path. Thus, $a$ has a charge of zero after stage~\ref{step:charging_1}. Hence, after stage~\ref{step:charging_2} $c$ has a
charge of at most $1$. If $c$ receives a nonzero charge at stage~\ref{step:charging_3}, then $c$ is an end node of a maximal (not directed) path in $G'$, and hence has degree $1$ in $G'$, showing that $c$ has a charge of $0$ after stage~\ref{step:charging_2} and a charge of $1$ after stage~\ref{step:charging_3}. 
\end{proof}

\begin{lemma}
\label{lemma:upper_bound_path_charge}
After the end of the charging scheme, every woman has a path-charge of at most $2$.
\end{lemma}

\begin{proof}
The plan of the proof is as follows. First, we prove that no woman receives a nonzero charge at both stage~\ref{step:charging_4} and stage~\ref{step:charging_5}. Then we prove that every woman receives charge of at most $2$ at stage~\ref{step:charging_4}, and finally
we prove that every woman receives a charge of at most $2$ at stage~\ref{step:charging_5}.

Let us suppose there was a woman $c\in B$ who receives a nonzero charge both at stage~\ref{step:charging_4} and stage~\ref{step:charging_5}. Since $c$ receives charge at stage~\ref{step:charging_4}, $c$ is an endpoint of a good path in $H$. Since $c$ receives charge at stage~\ref{step:charging_5}, there exists $b\in B$ such that $b$ is not matched by $M$ and such that $\mjumpe{b} = c$.  Since $c$ is matched by $M$ and $c$ is the end node of a good path, $c$ has degree $2$ in $G'$. Hence, $c \neq \mjump{b}$. Since $c=\mjumpe b$, we have that $\mjump{b}$ is in a $y$-node, which has no incoming arc in $H$. Moreover, there is a path in $H$
from $[\mjump{b}]$ to $c$, where all nodes except $c$ are $y$-nodes.  Note that every $y$-node has at most one incoming arc in $H$ and that $c$ has at most one incoming arc from a $y$-node. Thus, $[\mjump{b}]$ is in the good path containing $c$. However, no good path has a $y$-node with no incoming arc, contradiction.

Let us assume that a woman $c\in B$ gets a charge passed on from a woman $b\in B$ at stage~\ref{step:charging_4}. Thus, there is a critical arc $([a], [d])$, $a\in A$, $d\in B$ next to $b$, where
$\pathjump {[a]} = c$. By Lemma~\ref{lemma:pathjump_injective}, there is at most one
such critical arc $([a], [d])$. Also, since $a=M(b)$ there exists at most one woman $b$ who passes on her charge to $c$ at stage~\ref{step:charging_4}. By Lemma~\ref{lemma:upper_bound_unpopularity_charge}, such $b$ has a charge of at most $2$ after stage~\ref{step:charging_3}, showing that $c$ receives a charge of  at most $2$ at stage~\ref{step:charging_4}.

Let us assume that a woman $c\in B$ gets a nonzero charge passed on from a woman $b\in B$ at stage~\ref{step:charging_5}, where $\mjumpe b =c$.  Then $b$ is not matched by $M$, and hence $b$ has no charge after stage~\ref{step:charging_3}. By Lemma~\ref{lemma:critical_arc_is_in_some_good_path}, $b$ is in some good path in the graph $H$, because $b$ received a nonzero charge at stage~\ref{step:charging_4}. Hence, by  Lemma~\ref{lemma:mjumpe_injective} there is at most one such woman $b$ who passes her charge to $c$ at stage~\ref{step:charging_5}. Since $b$ has no charge after stage~\ref{step:charging_3} and receives a charge of  at most $2$ after stage~\ref{step:charging_4}, we know that $c$ receives a charge of at most $2$ at stage~\ref{step:charging_5}.
\end{proof}

\begin{lemma}
\label{lemma:upper_bound_total_charge}
At the end of the charging scheme, every woman has a charge of at most $3$.
\end{lemma}

\begin{proof}
Let $c\in B$ be some woman. If $c$ is not matched by $M$, the charge of $c$ is zero at the end of the charging scheme by Lemma~\ref{lemma:only_matched_women_have_charge}.  If $c$ is in a $5$-augmenting path, the charge of $c$ is zero at the end by Lemma~\ref{lemma:zero_total_charge_five_augmenting_path}. 

So we may assume that $c$ is matched by $M$, but is not in a $5$-augmenting path. Let $a\in A$ be the man such that $a = M(c)$. Obviously, $a$ is not in a $5$-augmenting path. Thus,
by Lemma~\ref{lemma:upper_bound_unpopularity_charge} $c$ has  an unpopularity-charge of at most $1$ at the end of the charging scheme.
By Lemma~\ref{lemma:upper_bound_path_charge}, $c$ has  a  path-charge of at most $2$ at the end, so the total charge of $c$ is at most $3$.
\end{proof}

\begin{theorem}
\label{thetheorem}
Let $\alpha_0 - \beta_0 - \dots - \alpha_k - \beta_k$  be an augmenting path of length at least $7$, where $\alpha_i\in A$ and $\beta_i\in B$, $i=0,1,\ldots,k$. Then there is a woman $\beta_i$, $i=1,\ldots,k-1$  who has zero unpopularity-charge at the end of the charging scheme.
\end{theorem}
\begin{proof}[Proof of Theorem~\ref{thetheorem}]
Let us assume that for every $i=1,\ldots,k-1$, the woman $\beta_i$ has a nonzero unpopularity-charge at the end of the algorithm.

For the proof we need the following definition.
For $i=1,\ldots,k$, we say that the man $\alpha_i$ is \tbdefined{pointing left} if $\alpha_i$
prefers $\beta_{i-1}$ to $\beta_i$; otherwise we say that $\alpha_i$ is \tbdefined{pointing right}.

\begin{claim}
\label{claim:pointing}
There is $i=1,\ldots,k-1$, such that $\alpha_i$ is pointing right. Moreover, $\alpha_k$ is pointing left.
\end{claim}

\begin{proof}
First, suppose for every $i=1,\ldots,k-1$, the man $\alpha_i$ is pointing left. 

In particular, $\alpha_1$ is pointing left, i.e. $\beta_0>_{\alpha_1} \beta_1$. Hence, we also have $\alpha_0 \geq_{\beta_0} \alpha_1$, since otherwise $(\beta_0,\alpha_1)$ is a blocking pair for $\opt$. Note, that $\alpha_0$ is $2$-promoted and has degree $1$ in $G'$, so $\beta_0$ rejected $\alpha_0$ as a $2$-promoted man. Because $\alpha_0 \geq_{\beta_0} \alpha_1$ and $(\alpha_1,\beta_0)$ is in $G'$,  the man $\alpha_1$ is $2$-promoted as well.

Furthermore, $\alpha_2$ is also pointing left, i.e. $\beta_1>_{\alpha_2} \beta_2$. Hence, we also have $\alpha_1 \geq_{\beta_1} \alpha_2$, since otherwise $(\beta_1,\alpha_2)$ is a blocking pair for $\opt$. Since $\alpha_1$ is $2$-promoted, $\beta_1$ rejected $\alpha_1$ as a $1$-promoted man. Because $\alpha_1 \geq_{\beta_1} \alpha_2$ and $(\alpha_2,\beta_1)$ is in $G'$, we have that the man $\alpha_2$ is promoted.

By the statement of the theorem, $\beta_1$ has a nonzero unpopularity-charge. Since during the algorithm $\beta_1$ rejected $\alpha_1$, the woman $\beta_1$ has $2$ proposals at the end of the algorithm. Hence, $\beta_1$ is not an end node of a maximal (not directed) path in $G'$ containing at least two distinct women, and thus $\beta_1$ did not receive any unpopularity-charge at stage~\ref{step:charging_3} of the charging scheme.

Thus, $\beta_1$ received a nonzero unpopularity-charge at stage~\ref{step:charging_2}. Let $a\in A$ be the man who passed on his charge to $\beta_1$ at stage~\ref{step:charging_2}. Because $a$ is in a $y$-node, $a$ is basic at the end of the algorithm. Since $\alpha_1$ and $\alpha_2$ are promoted, we know that $a$ is different from $\alpha_1$ and from $\alpha_2$. Since $(a,\beta_1)$ is an edge in $G'$ and $(\alpha_1,\beta_1)$ is not in $G'$ even though $\alpha_1$ is promoted and $a$ is basic, we have that $a >_{\beta_1} \alpha_1=\opt(\beta_1)$. On the other hand, since $a$ is in a $y$-node, we have that $\beta_1 >_a \opt(a)$, showing that $(a, \beta_1)$ is a blocking edge for $\opt$, contradiction.

Finally, let us show that $\alpha_k$ is pointing left. Note that $\beta_k$ is not matched by $M$. Thus, if $\alpha_k$ is pointing right, the edge $(\alpha_k,\beta_k)$ is blocking for $M$, contradiction.
\end{proof}

By Claim~\ref{claim:pointing}, there exists $\alpha_t$, $t=1,\ldots,k-1$, such that $\alpha_t$ is pointing to the right and $\alpha_{t+1}$ is pointing to the left. Let us fix such $t$. Note that $\alpha_t \indiff_{\beta_t}
\alpha_{t+1}$, because otherwise either $(\alpha_t,\beta_t)$ is blocking for $M$ or $(\alpha_{t+1},\beta_t)$ is blocking for $\opt$.

\begin{claim}
The woman $\beta_t$ is not $\alpha_t$-popular.
\end{claim}
\begin{proof}
By the assumption $\beta_t$ has a nonzero unpopularity-charge. Now, let us assume that $\beta_t$ received a nonzero unpopularity-charge at stage~\ref{step:charging_3}. Then $\beta_t$ has one proposal at the end of the algorithm, implying that $\beta_t$ is not $\alpha_t$-popular.

Now, let us assume that $\beta_t$ received a nonzero unpopularity-charge at stage~\ref{step:charging_2}. Let $a\in A$ be the man who passed on a nonzero charge to $\beta_t$ at stage~\ref{step:charging_2}. Since $a$ is in a $y$-node and $(a,\beta_t)$ is in $G'$, we have $\beta_t >_a \opt(a)$. Thus, we have $\alpha_t \geq_{\beta_t} a$, because otherwise $(a,\beta_t)$ is a blocking pair for $\opt$. 
Since $\beta_t$ is the less preferred neighbor of $a$ in $G'$, $\beta_t$ is not $a$-popular by Lemma~\ref{lemma:least_prefered_unpopular}. Furthermore, $\alpha_t \geq_{\beta_t} a$ and Lemma~\ref{lemma:popularity_monotonicity_preference}, imply that $\beta_t$ is not $\alpha_t$-popular.
\end{proof}

The next claim shows that there is a critical arc next to $\beta_t$. Thus, $\beta_t$ passes on all of her charge at stage~\ref{step:charging_4}, and so $\beta_t$ does not have any unpopularity-charge at the end of the charging scheme, finishing the proof of the theorem.

\begin{claim}
There is a critical arc next to $\beta_t$.
\end{claim}
\begin{proof}
Since $\beta_t$ is not $\alpha_t$-popular, by Lemma~\ref{lemma:monotonicity_popularity_over_time} $\beta_t$ rejected $\alpha_t$ at most once. Thus, there is $b\in B$ such that $(\alpha_t,b)$ is an edge in $G'$ and $b\geq_{\alpha_t}\beta_t=\opt(\alpha_t)$. Recall, that $\alpha_t$ is pointing right, i.e. $\beta_t>_{\alpha_t}\beta_{t-1}$, so $b\neq\beta_{t-1}=M(\alpha_t)$. Obviously, $\alpha_t$ is not in a $y$-node, so $(\alpha_t,b)$ is a critical arc next to $\beta_t$.
\end{proof}
\end{proof}

\begin{lemma}
\label{lemma:upper_bound_total_charge_augmentic_path}
At the end of the charging scheme, the sum of charges of all 
women in  a $(2\ell + 5)$-augmenting path is at most $3\ell$ .
\end{lemma}

\begin{proof}
Note that a $(2\ell + 5)$-augmenting path has $\ell + 3$ women. If $\ell = 0$, the statement follows by Lemma~\ref{lemma:zero_total_charge_five_augmenting_path}.
So we may assume $\ell \geq 1$.

The sum of the path-charges is at most $2\ell$. Indeed, the last woman has a zero path-charge by Lemma~\ref{lemma:only_matched_women_have_charge}. The first woman and second-to-last woman both have
a zero path-charge by Lemma~\ref{lemma:only_blue_women_have_path_charge}, since they are not blue women. Each of the remaining $\ell$ women has a path-charge of at most $2$ by Lemma~\ref{lemma:upper_bound_path_charge}.

The sum of the unpopularity-charges is at most $\ell$. Indeed, the last woman has a zero unpopularity-charge by Lemma~\ref{lemma:only_matched_women_have_charge}.  The first woman has a zero
unpopularity-charge by Lemma~\ref{lemma:x_node_no_unpopularity_charge}. By Theorem~\ref{thetheorem},
there exists a woman, who is not the first woman nor the last woman and who has a zero unpopularity-charge. None of the remaining $\ell$ women is matched by $M$ with a man from a $5$-augmenting path. Hence, each of the remaining $\ell$ women has an unpopularity-charge of at most $1$ by Lemma~\ref{lemma:upper_bound_unpopularity_charge}. Thus, the sum of charges of all women in the augmenting path is at most $3\ell$.
\end{proof}

\begin{theorem}
\label{theorem:approximation_guarantee}
The algorithm in~\cite{HuangKavitha2015} (described for the sake of completeness in Section~\ref{section:algorithm}) is a $\frac{13}{9}$-approximation algorithm.
\end{theorem}

\begin{proof}
Let $Q'$ be the set of all components of $M \cup \opt$.
For each $q \in Q'$, choose $\ell_q$ such that:
\begin{itemize}
\item If $q$ is an augmenting path, $q$ has $2\ell_q + 5$ edges.
\item If $q$ is a cycle or an even (non-augmenting)  alternating path, $q$ has $2\ell_q$ edges.
\item If $q$ is an odd (non-augmenting) alternating path, $q$ has $2\ell_q - 1$ edges.
\end{itemize}

Let $t$ be the number of $5$-augmenting paths and let $k$ be the number of augmenting paths in $M\oplus \opt$ which have at least $7$ edges. Let us also define $\ell_\Sigma:= \Sigma_{q \in Q'} \ell_q$.

\begin{claim}
\label{claim:main_theorem_inequalities}
We have $t \leq 3 \ell_\Sigma$ and $k \leq \ell_\Sigma$.
\end{claim}

\begin{proof}

Note that at the beginning of the charging scheme the sum of all charges equals $t$. Thus, to prove that $t \leq 3 \ell_{\Sigma}$ it is enough to show that the sum of all charges is at most $3 \ell_{\Sigma}$ at the end of the charging scheme.

We will prove  that for every
$q \in Q'$, the sum of the charges of the nodes in $q$ is at most $3 \ell_q$. If $q$ is an odd (non-augmenting) alternating path or an alternating cycle, then $q$ has at most $\ell_q$ women. Thus, the statement follows by Lemma~\ref{lemma:upper_bound_total_charge}. Similarly, if $q$ is an even (non-augmenting) alternating path, the statement follows by Lemma~\ref{lemma:only_matched_women_have_charge} and Lemma~\ref{lemma:upper_bound_total_charge}. 
If $q$ is an augmenting path, the statement follows directly from Lemma~\ref{lemma:upper_bound_total_charge_augmentic_path}.

Finally, we have $\ell_\Sigma \geq k$, because there are $k$ augmenting paths with at least $7$ edges and $\ell_q \geq 1$ for every such augmenting path $q$.
\end{proof}

\begin{claim}
\label{claim:main_theorem_output_size_optimum_size}
We have $\abs{M} \geq \ell_\Sigma + 2(t+k)$ and $\abs{\opt} \leq \ell_\Sigma +
3(t+k)$.
\end{claim}

\begin{proof}
First, $M$ has at least $\ell_\Sigma + 2(t+k)$ edges, since each cycle and each (non-augmenting) alternating path $q \in Q'$
has $\ell_q$ edges in $M$ and each augmenting path $q \in Q'$ has $\ell_q +
2$ edges in $M$.

Second, $\opt$ has at most $\ell_\Sigma + 3(t+k)$ edges, since  each cycle and each (non-augmenting) alternating path $q
\in Q'$ has at most $\ell_q$ edges in $\opt$ and each augmenting path $q \in
Q'$ has $\ell_q + 3$ edges.
\end{proof}

Thus, we have
\begin{align*}
\frac{\abs \opt}{\abs M}
&\leq
\frac{\ell_\Sigma + 3(t+k)}{\ell_\Sigma + 2(t+k)}
= \frac{13}{9} + \frac{-4\ell_\Sigma + (t+k)}{9(\ell_\Sigma + 2(t+k))}
\\ &\leq \frac{13}{9} + \frac{-4\ell_\Sigma + 3 \ell_\Sigma + \ell_\Sigma}{9(\ell_\Sigma + 2(t+k))}
= \frac{13}{9}\,,
\end{align*}
where the first inequality follows from Claim~\ref{claim:main_theorem_output_size_optimum_size} and the second inequality follows
from Claim~\ref{claim:main_theorem_inequalities} and from the fact that the denominator is non negative.
\end{proof}

For the sake of completeness, we give a proof that our analysis of the algorithm by Huang and Telikepalli~\cite{HuangKavitha2015}  is tight.
The next theorem follows from the example provided by Radnai~\cite{Radnai} (Section 3.6.1). 

\begin{theorem}
The bound in Theorem~\ref{theorem:approximation_guarantee} is tight.
\label{theorem:approximation_guarantee_tight}
\end{theorem}

\begin{proof}
The instance in Figure~\ref{figure: tight_instance} shows that the bound in Theorem~\ref{theorem:approximation_guarantee} is tight. In Figure~\ref{figure: tight_instance}, the circle nodes represent men and the square nodes represent women. The dashed lines represent the edges in $\opt$, the solid lines represent  the edges in $G'$, where the horizontal solid lines represent the edges in $M$.
The arrow tips indicate the preferences of each person. For example, $b_0^3$
has two tips arrows pointing to $a^2_2$ and $a_1^3$ and one tip arrow pointing to
$a_0^3$, so $b_0^3$ is indifferent between $a^2_2$ and $a_1^3$ and prefers both
to $a_0^3$.

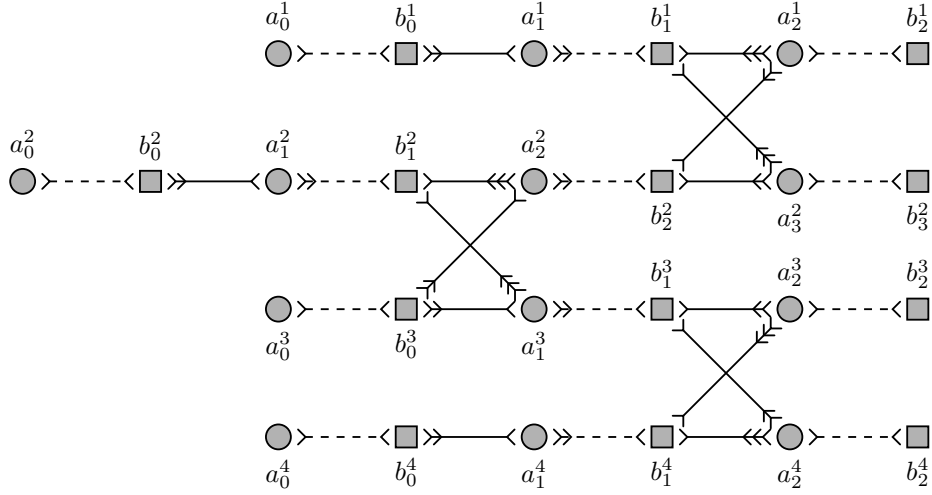
\begin{figure}
\begin{center}
\begin{tikzpicture}[scale=1.7]
\tikzstyle{every node}=[rectangle, fill = black!30!white, draw = black, minimum size = 8pt, line width = 0.7pt]
\tikzstyle{edge}=[> = angle 90, shorten >= 2pt, shorten < = 2pt, line width = 0.7pt]
\tikzstyle{man}=[circle]

\node[label=$a_0^1$, man] (f1a0) at (0, 0) {};
\node[label=$b_0^1$,    ] (f1b0) at (1, 0) {};
\node[label=$a_1^1$, man] (f1a1) at (2, 0) {};
\node[label=$b_1^1$,    ] (f1b1) at (3, 0) {};
\node[label=$a_2^1$, man] (f1a2) at (4, 0) {};
\node[label=$b_2^1$,    ] (f1b2) at (5, 0) {};

\node[label=$a_0^2$, man] (s1a0) at (-2, -1) {};
\node[label=$b_0^2$,    ] (s1b0) at (-1, -1) {};
\node[label=$a_1^2$, man] (s1a1) at (0, -1) {};
\node[label=$b_1^2$,    ] (s1b1) at (1, -1) {};
\node[label=$a_2^2$, man] (s1a2) at (2, -1) {};
\node[label=below:$b_2^2$,    ] (s1b2) at (3, -1) {};
\node[label=below:$a_3^2$, man] (s1a3) at (4, -1) {};
\node[label=below:$b_3^2$,    ] (s1b3) at (5, -1) {};

\node[label=below:$a_0^3$, man] (f2a0) at (0, -2) {};
\node[label=below:$b_0^3$,    ] (f2b0) at (1, -2) {};
\node[label=below:$a_1^3$, man] (f2a1) at (2, -2) {};
\node[label=$b_1^3$,    ] (f2b1) at (3, -2) {};
\node[label=$a_2^3$, man] (f2a2) at (4, -2) {};
\node[label=$b_2^3$,    ] (f2b2) at (5, -2) {};

\node[label=below:$a_0^4$, man] (f3a0) at (0, -3) {};
\node[label=below:$b_0^4$,    ] (f3b0) at (1, -3) {};
\node[label=below:$a_1^4$, man] (f3a1) at (2, -3) {};
\node[label=below:$b_1^4$,    ] (f3b1) at (3, -3) {};
\node[label=below:$a_2^4$, man] (f3a2) at (4, -3) {};
\node[label=below:$b_2^4$,    ] (f3b2) at (5, -3) {};

\draw[edge, dashed, >-<] (f1a0) -- (f1b0);
\draw[edge,         >>-<] (f1b0) -- (f1a1);
\draw[edge, dashed, >>-<] (f1a1) -- (f1b1);
\draw[edge,         >-<<<] (f1b1) -- (f1a2);
\draw[edge, dashed, >-<] (f1a2) -- (f1b2);

\draw[edge, dashed, >-<] (s1a0) -- (s1b0);
\draw[edge,         >>-<] (s1b0) -- (s1a1);
\draw[edge, dashed, >>-<] (s1a1) -- (s1b1);
\draw[edge,         >-<<<] (s1b1) -- (s1a2);
\draw[edge, dashed, >>-<] (s1a2) -- (s1b2);
\draw[edge,         >-<<] (s1b2) -- (s1a3);
\draw[edge, dashed, >-<] (s1a3) -- (s1b3);

\draw[edge, dashed, >-<] (f2a0) -- (f2b0);
\draw[edge,         >>-<] (f2b0) -- (f2a1);
\draw[edge, dashed, >>-<] (f2a1) -- (f2b1);
\draw[edge,         >-<<] (f2b1) -- (f2a2);
\draw[edge, dashed, >-<] (f2a2) -- (f2b2);

\draw[edge, dashed, >-<] (f3a0) -- (f3b0);
\draw[edge,         >>-<] (f3b0) -- (f3a1);
\draw[edge, dashed, >>-<] (f3a1) -- (f3b1);
\draw[edge,         >-<<<] (f3b1) -- (f3a2);
\draw[edge, dashed, >-<] (f3a2) -- (f3b2);

\draw[edge,         >>-<] (f1a2) -- (s1b2);
\draw[edge,         >>>-<] (s1a3) -- (f1b1);

\draw[edge,         >-<<] (s1a2) -- (f2b0);
\draw[edge,         >>>-<] (f2a1) -- (s1b1);

\draw[edge,         >>>-<] (f2a2) -- (f3b1);
\draw[edge,         >>-<] (f3a2) -- (f2b1);
\end{tikzpicture}
\end{center}
\caption{An instance for which the algorithm in~\cite{HuangKavitha2015} outputs a stable matching $M$ with $|M|= \frac{9}{13}\opt$ }
\label{figure: tight_instance}
\end{figure}

It is straightforward to check that $\opt$ is stable. Clearly, $M$ is a maximum matching in $G'$. We see that $\abs{M} = 9$ and $\abs{\opt} = 13$.

Let us verify that $M$ is indeed a matching output by the algorithm in~\cite{HuangKavitha2015}. For this it is enough to show, that the graph $G'$ can be obtained in the following run of the algorithm:
\begin{itemize}
\item $a_2^1$ proposes to $b_1^1$ twice; $b_1^1$ accepts twice.
\item $a^2_3$ proposes to $b_1^1$ twice; $b_1^1$ rejects his first proposal; for his second proposal, $b_1^1$ rejects one of $a_2^1$'s proposals.

\item $a^2_2$ proposes to $b^2_1$ twice; $b^2_1$ accepts twice.
\item $a_1^3$ proposes to $b^2_1$ twice; $b^2_1$ rejects his first proposal; for his second proposal, $b^2_1$ rejects one of $a^2_2$'s proposals.

\item $a_2^3$ proposes to $b_1^4$ twice; $b_1^4$ accepts twice.
\item $a_2^4$ proposes to $b_1^4$ twice; $b_1^4$ rejects his first proposal; for his second proposal, $b_1^4$ rejects one of $a_2^3$'s proposals.

\item $a_2^1$ proposes to $b^2_2$; $b^2_2$ accepts.
\item $a_3^2$ proposes to $b^2_2$; $b^2_2$ accepts.

\item $a_2^3$ proposes to $b_1^3$; $b_1^3$ accepts.
\item $a_2^4$ proposes to $b_1^3$; $b_1^3$ accepts.

\item $a_2^2$ proposes to $b^2_2$; $b^2_2$ chooses arbitrarily to reject $a_2^2$.
\item $a_1^3$ proposes to $b_1^3$; $b_1^3$ chooses arbitrarily to reject $a_1^3$.

\item $a_2^2$ proposes to $b_0^3$; $b_0^3$ accepts.
\item $a_1^3$ proposes to $b_0^3$; $b_0^3$ accepts.

\item $a_1^1$ proposes to $b_1^1$ twice; $b_1^1$ arbitrarily chooses to reject him twice.
\item $a_1^1$ proposes to $b_0^1$ twice; $b_0^1$ accepts twice.

\item $a_1^2$ proposes to $b^2_1$ twice; $b^2_1$ arbitrarily chooses to reject $a_1^2$ twice.
\item $a_1^2$ proposes to $b^2_0$ twice; $b^2_0$ accepts twice.

\item $a_1^4$ proposes to $b_1^4$ twice; $b_1^4$ arbitrarily chooses to reject $a_1^4$ twice.
\item $a_1^4$ proposes to $b_0^4$ twice; $b_0^4$ accepts twice.

\item $a_0^1$ proposes to $b_0^1$ several times; $b_0^1$ rejects him until $a_0^1$ gives up.
\item $a_0^2$ proposes to $b^2_0$ several times; $b^2_0$ rejects him until $a_0^2$ gives up.
\item $a_0^3$ proposes to $b_0^3$ several times; $b_0^3$ rejects him until $a_0^3$ gives up.
\item $a_0^4$ proposes to $b_0^4$ several times; $b_0^4$ rejects him until $a_0^4$ gives up.
\end{itemize}

Note that in this example, all men are either basic at the end of the algorithm, or are always rejected during the algorithm (even as $2$-promoted men). This means that simply adding more levels of promotion would not increase the algorithm's worst-case performance.
\end{proof}

\newpage
\bibliographystyle{plain}
\bibliography{bib}

\newpage
\section{Appendix}
\label{sec:appendix}

\subsection*{Proofs from Section~\ref{sec:tight_analysis}}

\begin{proof}[Proof of Lemma~\ref{lemma:three_path_original_paper}]
Let $\alpha_0-\beta_0-\alpha_1-\beta_1$ be a $3$-augmenting path in  $M\oplus \opt$, where $\alpha_0$, $\alpha_1\in A$ and $\beta_0$, $\beta_1\in B$. Since $\alpha_0$ is not matched by $M$, $\alpha_0$ is $2$-promoted. Since $\beta_1$ is not matched by $M$, $\alpha_1$ is basic. Moreover, $\alpha_1$ prefers $\beta_0$ to $\beta_1$, since otherwise $M$ has a blocking pair $(\alpha_1,\beta_1)$. Since $\alpha_1$ is basic, $\alpha_0$ is $2$-promoted and at the end of the algorithm $\beta_0$ has a proposal from $\alpha_1$,  we have that $\beta_0$ prefers $\alpha_1$ to $\alpha_0$, showing that $\opt$ has a blocking pair $(\alpha_1,\beta_0)$.
\end{proof}

\begin{proof}[Proof of Lemma~\ref{lemma:five_path_original_paper}]

First, $\alpha_0$ is not matched in $M$ and hence only one proposal of $\alpha_0$ is accepted at the end of the algorithm. Hence, $\alpha_0$ is $2$-promoted and $\alpha_0$ was rejected by $\beta_0$ as a $2$-promoted man.

Second, $\beta_2$ is not matched by $M$ and hence $\alpha_2$ was never rejected by $\beta_2$. Thus, $\alpha_2$ is basic. Moreover,  $\alpha_2$ prefers~$\beta_1$ to~$\beta_2$, since otherwise $M$ has a blocking pair $(\alpha_2,\beta_2)$.

Let us prove the first part of the third statement. If $\alpha_1$ is $2$-promoted, then $\alpha_1$ was rejected by $\beta_1$ as a $1$-promoted man. Since at the end of the algorithm $\beta_1$  has a proposal of the basic man $\alpha_2$, then   $\beta_1$ prefers $\alpha_2$ to $\alpha_1$. However, then $(\alpha_2,\beta_1)$ is a blocking pair for~$\opt$, contradiction.

Let us now prove the second part of the third statement. If $\alpha_1$ prefers $\beta_0$ to $\beta_1$, then $\beta_0$ does not prefer $\alpha_1$ to $\alpha_0$, since otherwise $(\alpha_1,\beta_0)$ is a blocking pair for $\opt$. Thus, either $\beta_0$ prefers $\alpha_0$ to $\alpha_1$ or  $\beta_0$ is indifferent between $\alpha_0$ and $\alpha_1$.  In both cases, it is impossible at the end of the algorithm to have the situation when  $\beta_0$ has a proposal of $\alpha_1$, while $\alpha_0$ has at most one his proposal accepted,  $\alpha_0$ is $2$-promoted and $\alpha_1$ is not $2$-promoted.

The forth statement follows from the fact that $\alpha_2$ prefers $\beta_1$ to $\beta_2$ and $\alpha_1$ prefers $\beta_1$ to $\beta_0$. Thus, if $\beta_1$ is not indifferent between $\alpha_1$ and $\alpha_2$, then either $\opt$ or $M$ is not stable.

For the fifth and sixth statements, it is enough to note that at the end of the algorithm $\alpha_1$ and $\alpha_2$ have both of their proposals accepted, while $\beta_0$ and $\beta_1$ hold two proposals.
\end{proof}

\begin{proof}[Proof of Lemma~\ref{lemma:no_yx_arc}(a)]
Suppose there is such an arc, say an arc $([a], [b])$, $a\in A$, $b\in B$. Since $a$ is in a $y$-node, $\opt(a)$ is not matched by $M$. Thus, $a$ was never rejected by $\opt(a)$. Hence, $a$ is basic and $a$ prefers $b$ over $\opt(a)$. 

Since $b$ is in an $x$-node, $\opt(b)$ is not matched by $M$. Thus, $\opt(b)$ is $2$-promoted. Moreover, $b$ prefers $a$ over $\opt(b)$, since $b$ accepted a basic proposal of $a$ while rejecting a $2$-promoted proposal of $\opt(b)$. This implies that $(a,b)$ is a blocking edge for the stable matching $\opt$, contradiction. 
\end{proof}

\begin{proof}[Proof of Lemma~\ref{lemma:critical_arc}(b)]
Due to the definition of a critical arc, no critical arc starts at a $y$-node. Let us show that no critical arc ends at an $x$-node. Suppose there is a critical arc $([a],[b])$, $a\in A$, $b\in B$ such that $b$ is in an $x$-node. 

Since $b$ is in an $x$-node, $\opt(b)$ is not matched by $M$, and hence $\opt(b)$ is $2$-promoted. Moreover, by the definition of a critical arc, $a$ is not $2$-promoted. Thus, $b$ prefers $a$ over $\opt(b)$, since $b$ accepted a not-2-promoted proposal of $a$ while rejecting a $2$-promoted proposal of $\opt(b)$. 

By the definition of a critical arc, $b\geq_a \opt(a)$. Note that $b$ is not equal to $\opt(a)$, because $\opt(b)$ is $2$-promoted and $a$ is not $2$-promoted. Hence, $a$ prefers $b$ over $\opt(a)$. This implies that $(a,b)$ is a blocking edge for the stable matching $\opt$, contradiction.  
\end{proof}

\begin{proof}[Proof of Lemma~\ref{lemma:critical_arc_is_in_some_good_path}(c)]

Before showing that every critical arc is contained in some good path, let us make some observations about $x$-nodes and $y$-nodes in the graph $H$.

In the graph $H$, every $y$-node is either isolated or has an outgoing arc in $H$: If the $y$-node consists of $a\in A$ and $b\in B$, then $a$ is basic and hence at the end of the algorithm both his proposals are accepted. If these two proposals are accepted by the same woman, namely $b$, then the $y$-node is isolated. Otherwise, the $y$-node has an outgoing arc.

Similarly, every $x$-node is either isolated or has an incoming arc. Indeed, if the $x$-node consists of $a\in A$ and $b\in B$, then $b$ rejected $\opt(b)$ at least once, and hence at the end of the algorithm $b$ holds two proposals. If these two proposals are both from the same man, namely $a$, then the $x$-node is isolated. Otherwise, the $x$-node has an incoming arc.

Now, let us consider a critical arc $([a],[b])$, $a\in A$, $b\in B$ in the graph $H$. Let us construct a good path in $H$ by constructing a subpath from $[b]$ to a blue woman and a subpath from a blue man to $[a]$. If $[b]$ is a blue woman, then the subpath from $[b]$ to a blue woman is a trivial path. Otherwise, $[b]$ is a $[y]$-node, since by Lemma~\ref{lemma:critical_arc}(b), $[b]$ is not an $x$-node. Note that $[b]$ is not isolated and thus $[b]$ has an outgoing arc. Moreover, by Lemma~\ref{lemma:no_yx_arc}(a) this outgoing arc is leading to either a blue woman or a $y$-node. In case this arc leads to a blue woman, we stop, otherwise we continue until we reach a blue woman in $H$. 

We construct analogously a subpath from a blue man to $[a]$. From the construction it is clear that the path is a good path.
\end{proof}

\begin{proof}[Proof of Lemma~\ref{lemma:good_path_has_at_most_one_critical_arc}(d)]
Note that all the internal nodes of a good path must be $x$- or $y$-nodes. To see this, assume that
$c$ is an internal blue node of a good path. Then, $c$ is incident to two distinct edges in $H$, none of which is in $M$, contradicting the assumption that $M$ matches all degree two nodes in the graph $H$.
The statement then follows from Lemma~\ref{lemma:no_yx_arc}(b) and Lemma~\ref{lemma:critical_arc}(c).
\end{proof}

\begin{proof}[Proof of Lemma~\ref{lemma:good_paths_are_node_disjoint}(e)]


First, note that two distinct good paths cannot have common end nodes.
Let us assume that $c$ is a common end node of two distinct good paths. So $c$ is incident to two distinct edges in $H$, none of which is in $M$, contradicting the assumption that $M$ matches all degree two nodes in the graph $H$.

Now recall that all nodes in $H$ have degree at most $2$. If two distinct good paths have a common internal node, since the end nodes are different,
then necessarily there is a node with degree $3$ in $H$, a contradiction.

\end{proof}

\end{document}